\documentclass[journal,12pt,onecolumn,draftclsnofoot,]{IEEEtran}
\usepackage[T1]{fontenc} 
\usepackage[english]{babel} 
\usepackage{amsmath,amsfonts,amsthm,amssymb} 
\usepackage{algorithm}
\usepackage{enumitem}
\usepackage[noend]{algpseudocode}
\usepackage{cite,graphicx,float}
\usepackage{subcaption}
\usepackage{color}
\usepackage{caption}
\usepackage{soul}
\usepackage{ulem}
\usepackage[figurewithin=none]{caption}

\usepackage{fancyhdr} 
\newtheorem{theorem}{Theorem}
\newtheorem{lemma}{Lemma}

\newcommand{\squeezeup}{\vspace{-3mm}}

\title{	\Large On-Request Wireless Charging and Partial Computation Offloading In Multi-Access Edge Computing Systems\\ 
}

\author{Rafia Malik and Mai Vu\\
\small Department of Electrical and Computer Engineering, Tufts University, MA, USA\\
Email: rafia.malik@tufts.edu, mai.vu@tufts.edu} 

\date{\normalsize{May 21, 2018}} 

\begin{document}
\maketitle 

\begin{abstract}
Wireless charging coupled with computation offloading in edge networks offers a promising solution for realizing power-hungry and computation intensive applications on user devices. We consider a multi-access edge computing (MEC) system with collocated MEC server and base-station/access point, each equipped with a massive MIMO antenna array, supporting multiple users requesting data computation and wireless charging. The goal is to minimize the energy consumption for computation offloading and maximize the received energy at the user from wireless charging. The proposed solution is a novel two-stage algorithm employing nested primal-dual and linear programming techniques to perform data partitioning and time allocation for computation offloading and design the optimal energy beamforming for wireless charging, all within MEC-AP transmit power and latency constraints. Algorithm results show that optimal energy beamforming significantly outperforms other schemes such as isotropic or directed charging without beam power allocation. Compared to binary offloading, data partition in partial offloading leads to lower energy consumption and more charging time, and hence offers better wireless charging performance. The charged energy over an extended period of time both with and without computation offloading can be substantial. Opportunistic wireless charging from MEC-AP thus offers a viable untethered approach for supplying energy to user-devices.
\end{abstract}

\section{Introduction}
Multi-Access Edge Computing (MEC) networks have recently garnered significant interest thanks to its ability to provide cloud-computing capabilities within the radio access network, offering proximity, low latency, and high rate access. MEC can bring computing intensive features such as augmented and virtual reality to a large number of connected wireless devices with limited processing capability and battery lifetime by providing services such as computation offloading and wireless charging. Future generation networks offer native support for edge computing functionality, such as key enablers defined by the 3GPP in 5G system architecture to support edge computing \cite{23.501}. A typical deployment scenario is where the MEC server is co-located with the base-station/access-point (BS/AP)~\cite{MEC2018}. At the same time, the exponentially growing number of connected devices leads to network densification with a large number of deployed APs. With multiple MEC-APs deployed over a relatively small area in close vicinity to the connected users, RF wireless power transfer from the APs to the user devices becomes practical. 

Far-field wireless power transfer using Radio Frequency (RF) enables energy-constrained devices to replenish their charge levels without physical connections, offering the inherent advantage of untethered mobility and battery sustainability~\cite{Malik2018}. There has been significant recent progress in wireless power transfer technology ranging from battery-free cellphone operating on harvested energy from RF signals transmitted by a BS 31 feet away~\cite{Smith2017} to reconfigurable RF rectifiers capable of handling the variable nature of input power at the energy harvesting circuits~\cite{Sanchez2017}. Commercial products employing RF power transfer have also appeared on the market, charging multiple devices up to 15 meters away ~\cite{Powercaster}\cite{Ossia}\cite{Energous}. Wireless power transfer in future systems is expected to charge devices at distances ranging from a few meters (for example smart phones) to hundreds of meters (for example sensors)~\cite{Zeng2017}. Adding wireless charging to MEC networks as an \textit{on-request} feature can further help in achieving the required availability and reliability of energy supply, which has become crucial for today's QoS-sensitive applications~\cite{Ekram2016}. 

Prior works have considered the symbiotic convergence of edge computing and wireless power transfer in different deployment scenarios, for example, wireless charging in cooperation assisted edge computing~\cite{Hu2018}, UAV-enabled mobile edge computing~\cite{Chu2018} and MEC based heterogeneous networks~\cite{Ji2018}. Wireless power transfer has been considered in MEC networks for \textit{self-sustained} devices, which rely on wireless charging as their sole power source, in relay-aided edge systems~\cite{Hu2018}, single user~\cite{Chae2016} and multiple user systems~\cite{Wang2018}. Such scenarios are typical for devices with low power requirements and/or low receiver sensitivity. Significantly different from this, an \textit{on-request} wireless charging model is where each user-terminal has its own power source and can use wireless charging from the AP to supplement its power consumption. Such \textit{on-request} charging schemes can minimize the associated energy costs of power transfer and are likely to become an integral part of the maturing 5G vision in the near future~\cite{Ekram2016}.

For multiuser edge networks, the transmission strategy and multiple access scheme can significantly impact the overall latency. In terms of communication and data transfer, existing works typically employ sequential protocols like Time Division Multiple Access (TDMA)~\cite{Bi2018}\cite{Wang2018}\cite{Chae2016}\cite{You2017}. Instead, massive MIMO enables simultaneous data offloading from multiple users to the MEC-AP and hence dramatically reduces the wireless transmission time. Employing massive MIMO at the MEC-AP also delivers high throughput and energy efficiency with transmit power savings because of beamforming gains. Massive MIMO can reduce the transmit power at the AP for a given data rate and therefore also has a positive impact on the system energy consumption. In terms of wireless charging, having a large number of antennas at the MEC-AP leads to increased charging range since a larger amount of energy can be reliably directed and transferred~\cite{Kashyap2016}\cite{Amarasuriya2016}. Hence massive MIMO technology can prove to be highly effective for energy beamforming for efficient charging of the user devices. Prior works only consider wireless charging from MEC servers where the AP is equipped with single antenna~\cite{Bi2018}\cite{Hu2018}, or having multiple antennas but not with massive MIMO capability~\cite{You2017}\cite{Wang2018}. Massive MIMO can be deemed an enabling technology for wireless charging because of its ability to focus energy via sharp beams and charge multiple users concurrently.

In this work, we consider a multi-cell multi-user network scenario where access points equipped with massive MIMO antenna arrays and with co-located mobile edge computing servers offer computation offloading and wireless charging. This model generalizes several existing problems considered in literature on edge computing systems by integrating massive MIMO and power transfer features, which to our knowledge is the first to do so. In our proposed system model, we integrate two different services, computation offloading and wireless charging, which are independent of each other in terms of operation but are bounded by the same latency and power constraints. In our proposed formulation, the objective of the data offloading problem is to minimize the amount of consumed energy, while the objective of energy harvesting is to maximize the received energy. We therefore treat the two problems of offloading and charging sequentially as primary and secondary, where the primary problem aims at energy efficiency for joint communication and computation and the secondary problem aims at achieving a best-effort solution (within the available time slot without violating the latency constraint) for maximum wireless transferred energy. This is different from our follow-up work in \cite{Rafia2020}, where we consider a joint optimization of both computation offloading and wireless charging with the common goal of minimizing the amount of consumed energy, however at the price of a reduced overall received (charged) energy.

Different from self-sustained model which is usually restricted to low-power passive sensors and wearable devices \cite{Ismail2018}, here we propose a system model applicable to an active-user use-case, for instance inside a sports stadium, where multiple smart phone users may request computation data offloading and/or wireless charging. While computation offloading may be needed for AR/VR applications in such a use-case, wireless charging is a complementary billable service provided to further enhance the user experience. The computation offloading service is often time critical (for example, due to an upper bound on the motion-to-photon latency for AR/VR applications \cite{MEC2018}), and therefore offloading requests by the users must be met within the current time block, leading to the latency constraint. On the other hand, wireless charging is opportunistic, and while the MEC-AP may start fulfilling energy requests arriving in any time block, these requests can be accrued and carried over to the next time blocks for fulfillment since charging is not as time-sensitive as computation offloading.

Such a system model has versatile applicability to different use-cases. Examples include (i) AR/VR applications in human-machine interfaces used in smart factories, where complex processing tasks may be offloaded to the edge network, which not only enables easy access to different context information available in the network but also prevents head-mounted AR/VR gear from becoming too warm and uncomfortable to wear \cite{22.104}, (ii) gaming or training service data between two 5G connected devices \cite{22.261}, (iii)  real-time map rendering for autonomous vehicular applications \cite{Malik2020}, and (iv) professional low-latency periodic audio transport services for Audio-Video (AV) production applications, music festivals etc. \cite{22.263}.

\subsection*{Major Contributions}
The main contributions of this work can be summarized as follows.

\begin{enumerate}
\item We propose a system model that integrates two independent MEC services of computation offloading and wireless charging in the same MEC system under the same set of constraints on latency and transmit power. We treat these two problems separately with independent objectives, one of minimizing the energy consumption for computation offloading, and the other of maximizing the received (charged) energy for wireless charging. The two problems, however, are coupled together via system latency constraint where wireless charging is performed opportunistically during each computation offloading period but can span over multiple periods to satisfy as much of the requested charging amount as feasible.

\item We formulate novel problems to minimize the energy consumption in computation offloading, and to maximize the received energy at the users end in wireless charging. We design novel and efficient sequential algorithms to solve these problems. The first algorithm optimizes the data partitioning, transmit power for wireless transmission and time allocation through a nested-structure using a latency-aware descent algorithm \cite{Malik2020} and a primal-dual algorithm. The derived optimal time allocation is then fed to a second algorithm which finds the optimal energy beamforming matrix (including beam power allocation and beam directions) through another nested structure using a primal dual algorithm and linear programming. The proposed sequential algorithms inherently prioritize computation offloading during the resource allocation process, such that the \textit{on-request} wireless charging functionality is only enabled in the current time block if latency and power constraints permit. This is a key practicality feature where precedence is given to time-critical computation offloading at the MEC and wireless charging requests may be fulfilled as necessary in multiple time blocks.

\item Using our proposed algorithm, we provide detailed quantitative performance analysis and study the impact of different system parameters and optimizing variables on the energy consumption and wireless charging performance. We show that data partitioning is a key variable affecting system energy consumption, while latency is paramount for wireless charging performance. Our proposed solution allows continuous wireless charging over an extended period of time (including multiple time periods/blocks) both with and without computation offloading, as long as the device is connected to the MEC-AP. These theoretical and numerical analyses work can serve as sound guidance for practical implementation. 
\end{enumerate}

\subsubsection*{Notation} $\boldsymbol{X}$ and $\boldsymbol{x}$ denote a matrix and vector respectively, $\nabla^2 f(x)$ denotes the Hessian matrix, and $\nabla^2 f(x)^{-1}$ denotes its inverse. For an arbitrary size matrix, $\boldsymbol{Y}$,  $\boldsymbol{Y}^\ast$ denotes the Hermitian transpose, and $\textbf{diag}(y_1,...,y_N)$ denotes an $N\times N$ diagonal matrix with diagonal elements $y_1,...,y_N$. $\boldsymbol{I}$ denotes an identity matrix, and $\boldsymbol{0, 1}$ denote an all zeros and all ones vector respectively. The standard circularly symmetric complex Gaussian distribution is denoted by $\mathcal{CN}(\boldsymbol{0}, \boldsymbol{I})$, with mean $\boldsymbol{0}$ and covariance matrix $\boldsymbol{I}$. $\mathbb{C}^{k \times l}$ and $\mathbb{R}^{k \times l}$ denote the space of $k \times l$ matrices with complex and real entries, respectively.

\section{System Model}\label{sys_model}
We consider a system where $L \geq 1$ Access Points (APs), each co-located with an MEC Server, are deployed over a targeted zone/area, for instance in a sports stadium or a town fair, serving ground users with computation offloading and power transfer. Each AP is equipped with a massive antenna array with $N$ antennas while the user-devices are equipped with single antennas. These APs wirelessly charge (upon request) ground users in downlink, collect offloaded data from the users in uplink, and deliver computed results to users in downlink~\cite{MEC2014}. We consider $K$ users requesting wireless charging service and sending data for computation offloading to each MEC-AP. In the case of cellular networks, wireless charging can be a billable service assuming that the ground users have knowledge of their battery state, and can inform the AP about their battery level for requesting recharge when their battery is critically low.

\setlength{\belowcaptionskip}{-20pt}
\begin{figure}[t]
\centering
\includegraphics[scale = 0.65]{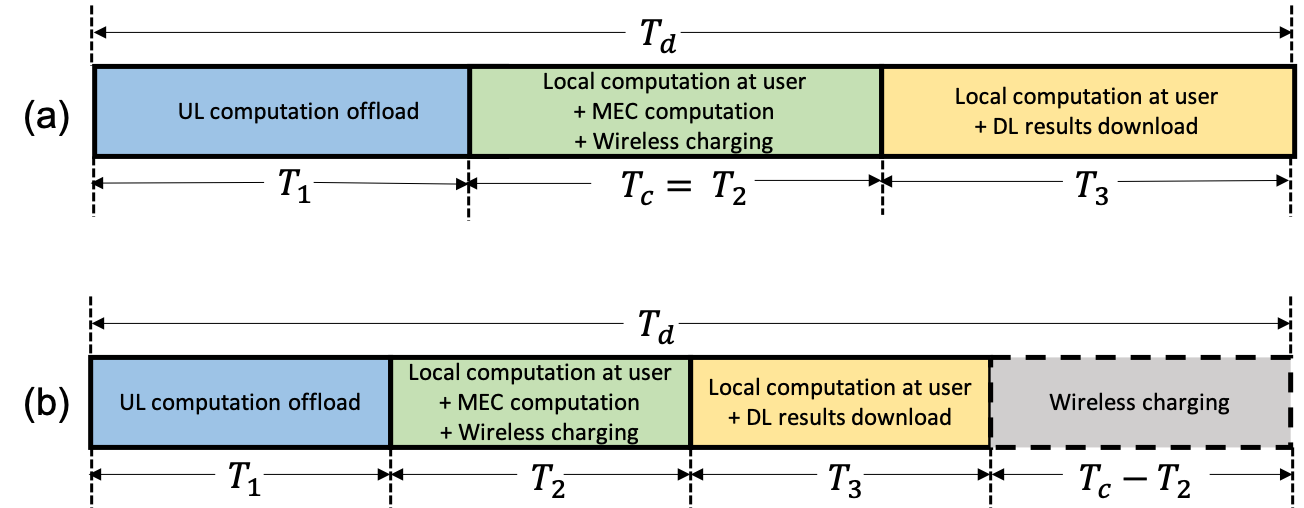}
\caption{Timing diagram and functional model of the system's operation for both data computation offloading and wireless charging. Wireless charging is performed opportunistically during MEC computation period when latency is tight (a), and also after computation offloading finishes when latency is not tight (b).}
\label{fig:phases}
\end{figure}


Consider the case where wireless charging is requested jointly with computation offloading, which includes the scenario of charging only or computation only as special cases. There are three functions contributing to the system's operation as shown in Figure~\ref{fig:phases}; (i) wireless charging of the user terminals by the MEC-AP, (ii) data transmission in the form of computation offloading from the users to the MEC-AP in the uplink and results downloading from the MEC-AP to the users in the downlink, and (iii) data computation at the MEC server and locally at the users. 

Given a latency constraint of $T_d$, the time span for data offloading, computation at both the users and the MEC ends, wireless charging, and delivery of computed results to the user must not exceed $T_d$. Considering computation offloading, this operation is divided into three timing phases: The time duration for data offloading to the MEC is denoted by $T_1$, the computation for offloaded data at the MEC spans duration $T_2$, and the transmission of processed results occupies time $T_3$. The timing during for wireless charging will be dependent on these three computation offloading phases and the total latency. Figure~\ref{fig:phases} shows two scenarios timing model: either computation offloading requires the whole duration of $T_d$, in which case the wireless charging is restricted to the computation phase, or computation offloading consumes a time duration less than $T_d$ and therefore wireless charging can continue after computed results have been transmitted in downlink. Our formulations in the next section account for both of these scenarios. We discuss the energy and time consumption of each system's function, namely wireless charging, data transmission and data computation.

\subsection{Wireless Charging}
\setlength{\belowcaptionskip}{-20pt}
\begin{figure}[t]
                \centering
                \includegraphics[scale = 0.575]{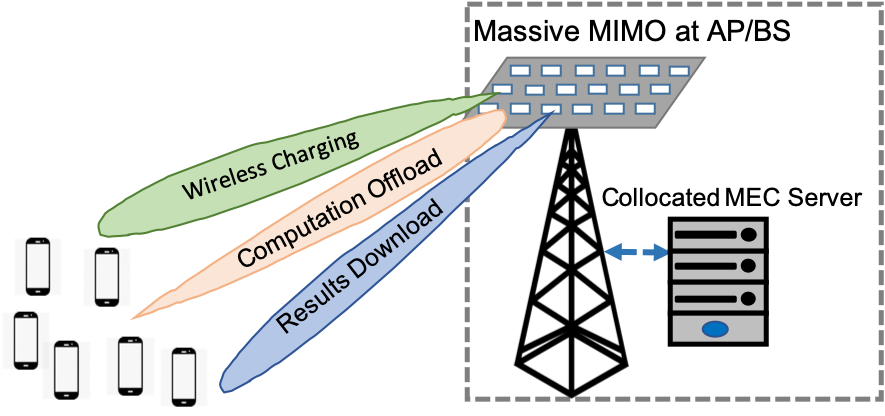}
                \caption{Beamforming for maximal information and energy transfer using massive MIMO antenna array at AP}
                \label{fig:beam}
\end{figure}
In each cell, we consider $K$ users requesting wireless charging from the MEC-AP, where the $i^{\text{th}}$ user requests $e_i$ mJ of energy. To cater for the energy requests from multiple users, the massive-MIMO enabled MEC-AP employs transmit energy beamforming, as shown in Figure~\ref{fig:beam}. Such energy beamforming requires channel state information (CSI), which can be obtained at the AP using \textit{uplink training}, where pilot symbols are transmitted over some duration of the coherence interval to estimate the channel matrix from the users to their serving MEC-AP. For the downlink channel, we assume Time Division Duplex (TDD) operation such that the channel matrix from the AP to the users can be obtained by wireless channel reciprocity of the uplink channel and hence the transmission of downlink pilots becomes unnecessary\cite{Marzetta2016}\cite{Ngo2017}.

Let $\boldsymbol{x_q}$ denote the energy bearing signal from the AP to the user-terminal (UT), $\boldsymbol{W_q} \triangleq \mathbb {E}\Big[\|\boldsymbol{x_q}\|^2\Big]$ denote the transmit covariance matrix, and $P_c = \text{tr}(\boldsymbol{W_q})$ be the power transmitted from the AP for wireless charging, in short, the charging power. Then the received (charged) power at the $i^{\text{th}}$ user is given as
\begin{equation}
P_{h,i} = \xi_i \mathbb{E} \left [ \left | \boldsymbol{h_i^\ast x_q}\right |^2\right ] = \xi_i \text{tr} (\boldsymbol{h_i^\ast W_q h_i})
\end{equation}
where $0 \leq \xi_i \leq 1$ is the energy conversion efficiency from Radio Frequency (RF) to Direct Current (DC) for the $i^{\text{th}}$ user and $\boldsymbol{h_i} \in \mathbb{C}^{N \times 1}$ is the channel from the AP to the $i^{\text{th}}$ user. We assume a linear energy harvesting model where the energy conversion efficiency per user is constant over a single time block duration, $T_d$. Non-linear wireless charging models, with a variable energy conversion efficiency over time, are more applicable to scenarios where there is a variation in the received power\cite{Schober2017} such as at high SNR and also depends on the rectifier characteristics (diode breakdown region) \cite{Bruno2017}. For the considered system with strict latency constraint on each time block for on-request charging, constant user energy conversion efficiency is more suitable. To account for the difference in received power at each user location, each $i^{\text{th}}$ user has its own energy conversion efficiency $\xi_i$ based on the received power in the current time block.

We define $T_c$ as the time duration for wireless charging, where $T_c = T - (T_1 + T_3)$ and includes the time consumed by the computation phase, over which power is transferred to the users alongside computation at the users and at the MEC server. The energy consumed at the MEC server for power transfer, in short the charging energy, is given by
\begin{equation}\label{Ec}
E_c = T_c \text{tr}(\boldsymbol{W_q})
\end{equation}

We consider an \textit{opportunistic wireless charging maximization} approach where the received energy at the users is maximized subject to the latency and MEC-AP's transmit power constraint. For the $i^{\text{th}}$ user requesting $e_i$ amount of energy, the received (charged) energy, $E_{h,i}$, is constrained as below
\begin{equation}\label{cond_energy}
E_{h,i} = P_{h,i} T_c = \xi_i T_c \text{tr} (\boldsymbol{h_i^\ast W_q h_i}) \leq e_i \ \forall \ i \in [1,K]
\end{equation}
Here the amount of wireless charging is upper-bounded by $e_i$ such that the charged energy is at most equal to the requested amount so as not to overcharge the users since charging is a billable service, and also not to burn the user's battery. Having this bound ensures feasibility of energy transfer. In this way no single user gets an unfairly large amount of the charged energy at the expense of others, and only a portion of the requested energy may be charged (in the current time block) if it is unfeasible for the AP to satisfy the user's energy request completely due to poor channel conditions or high energy request(s) by a single or few users. 

Note that in cases where the $i^{\text{th}}$ user's energy request is only partially fulfilled in the current time block, the remaining amount may be charged in subsequent time blocks. Since, in our considered system model, the operation of computation offloading is not dependent on wireless charging for energy, charging can be deferred to future time blocks if computation offloading demands more time and energy resources in the current time block. For the remaining of the paper, to simplify notation, we assume that for the current time block, no amount of energy has previously been received by the user, and the charge requested by the $i^{\text{th}}$ user is equal to $e_i$. All subsequent formulations and algorithms, however, are applicable if this requested energy is scaled by a factor to reflect a proportion in each time block.

\subsection{Data Transmissions}
For computation offloading at each MEC, we consider the simple \textit{data-partition model}, where the task-input bits are bit-wise independent and can therefore be arbitrarily divided into different groups to be executed by different entities~\cite{Mao2017}. We consider the case of partial offloading, such that for the $i^{th}$ user, the $u_i$ computation bits are partitioned into $q_i$ and $s_i$ bits, where $q_i$ bits are computed locally and $s_i$ bits are offloaded to the MEC server. Assuming that such partition at the user-terminal does not incur additional computation bits, then $u_i = q_i + s_i$.
\subsubsection{Offloading Data in Uplink}
In a given time slot, $K$ single-antenna user terminals simultaneously offload to the $N$ antenna AP. We consider $N \gg K$ such that the throughput becomes independent of the small-scale fading with channel hardening~\cite{Ngo2017}. The very large signal vector dimension at a massive MIMO AP enables the use of linear detectors such as maximum ratio combining (MRC), in which case the uplink net achievable transmission rate for the $i^{th}$ user in the $l^{th}$ cell, $r_{u,i}$, is given as~\cite{Marzetta2016}
\begin{equation}\label{rate_ul}
r_{u,i} = \nu \log_2 \left ( 1 + \frac{\text{SINR}_{li}^{ul}}{\Gamma_{1}} \right ), \ \text{SINR}_{li}^{ul} = \frac{N \gamma_{li}^l p_{li}}{\sigma_{1,li}^2}
\end{equation}
where $\Gamma_{1} \geq 1$ accounts for the capacity gap due to practical coding schemes, $\gamma_{li}$ is the mean-square channel estimate, and $p_{li}$ is the transmit power of the $i^{\text{th}}$ user in the $l^{\text{th}}$ cell.  The constant $\nu$ represents the portion of transmission symbols spent on data transfer in the coherence interval $\tau_c$. The interference and noise power, $\sigma_{1,li}^2$, includes the receiver noise variance, interference due to channel estimation and from contaminating cells, and inter-cell interference as defined in~\cite[Eq. 4.18]{Marzetta2016}, and is dependent on all users' transmit power and channel conditions \cite{Malik2020}.

The energy consumed for offloading the $i^{th}$ user's data is given by $E_{OFF,i} = p_i t_{u,i}$, where $p_i$ is the transmit power and $t_{u,i}$ is the transmission time for the $i^{th}$ user. Let $B$ denote the channel bandwidth, then $t_{u,i} = \frac{s_i}{B r_{u,i}}$. All users offload their computation bits simultaneously, and the total energy and time overhead for simultaneous data offloading is given as
\begin{equation}\label{E_ul}
E_{OFF} = \sum_{i=1}^{K} \frac{p_i s_i} {B r_{u,i}}, \ T_1 = \max_{i \in [1,K]} t_{u,i}.
\end{equation}

\subsubsection{Downloading Results in Downlink}
For the $i^{th}$ user in the $l^{th}$ cell, the downlink transmission rate with maximum ratio linear precoding at the MEC-AP is given as~\cite{Marzetta2016}
\begin{equation}\label{rate_dl}
r_{d,i} = \log_2 \left ( 1 + \frac{\text{SINR}_{li}^{dl}}{\Gamma_{2}} \right ), \ \text{SINR}_{li}^{dl} = \frac{N P \gamma_{li}^l \eta_{lk}}{\sigma_{2,li}^2}
\end{equation}
where $\Gamma_{2} \geq 1$ is the capacity gap, and $\sigma_{2,li}^2$ is the interference and noise power which also contains pilot contamination and intercell interference as given in~\cite[Eq. 4.34]{Marzetta2016}, and depends on the power allocation at the MEC-AP for downlink wireless transmission and also on the channels between the AP and the users \cite{Malik2020}.

The transmission time for delivering the $i^{th}$ user's computation results can be written in terms of the downlink rate in (\ref{rate_dl}) as $t_{d,i} = \frac{\tilde{s}_i}{B r_{d,i}}$. Here $\tilde{s}_i$ denotes the number of information bits generated after processing $s_i$ offloaded bits of the $i^{th}$ user. The number of information bits generated as a result of data computation ($\tilde{s}_i$) are proportional to the data bits to be computed ($s_i$), that is $\tilde{s}_i \propto s_i \to \tilde{s}_i = \mu s_i$. $\mu$ is the proportionality parameter between the amounts of requested and computed data and is not restricted to the range [0,1], rather it adds an application-centric flexibility to our system model in terms of the data size in downlink. For instance, $\mu < 1$ for face recognition applications or $\mu \gg 1$ for video-rendering applications \cite{Chen2015}\cite{Mangiante2017}\cite{Malik2020}. The AP simultaneously transmits computed results for all users, and the total energy and time overhead for results downloading are then given as
\begin{equation}\label{E_dl}
E_{DL} = \sum_{i=1}^{K} \frac{P \eta_i \mu s_i}{B r_{d,i}}, \ T_3 = \max_{i \in [1,K]} t_{d,i}.
\end{equation}

\subsection{Data Computation}
\subsubsection{Local computation at the users}
The time for computation depends on the amount of data to be computed and the CPU cycle frequency. The energy consumption and the processing time for local computation at the $i^{th}$ user is given as~\cite{Mao2017} 
\begin{align}\label{t_Li}
&E_{LC} = \sum_{i=1}^{K} \kappa_i c_i (u_i - s_i) f_{u,i}^2, \ \ t_{L,i} = \frac{c_i (u_i - s_i)}{f_{u,i}}
\end{align}
where $\kappa_i$ is the effective switched capacitance, $f_{u,i}$ denotes the average CPU frequency, $c_i$ denotes the CPU cycle information, and $q_i = u_i - s_i$ is the total number of bits required to be locally computed at $i^{th}$ user respectively. While it is possible for the user to offload data and perform local computation at the same time, the user's power is limited, thus as a common assumption \cite{Chae2016}\cite{Wang2018}, the user device focuses its power for offloading data and performs no computation during Phase I. The users' local computation time starts in Phase II and can also extend to Phase III while the MEC is sending computed results back to users. This fact is considered later in the problem formulations.

\subsubsection{Computation of the offloaded data at the MEC server}
MEC servers, with high computation capacities, compute the tasks of all users in parallel~\cite{Taleb2017}\cite{Mao2017}. The energy and time consumed for computing offloaded bits is given as
\begin{equation}\label{tMEC}
E_{OC} = \sum_{i = 1}^{K} \kappa_m f_{mi}^2 d_m s_i, \ \ t_{M,i} = \frac{d_m s_i}{f_{mi}} \ \forall i \in [1, K], \ \ T_2 = \max\{t_{M,i}\}.
\end{equation}
where $t_{M,i}$ is the time for computing $i^{th}$ user's offloaded task, $s_i$ is the number of bits offloaded by the $i^{th}$ user to the MEC, $d_m$ is the number of CPU cycles required to compute one bit at the MEC, $f_{mi}$ is the CPU frequency assigned to the $i^{th}$ user's task, and $\kappa_m$ is the effective switched capacitance of the MEC server. The computation at the MEC is synchronous such that computation only begins after data from all users has been offloaded. While it is possible to perform fine-scale timing optimization where the MEC starts computing immediately after it receives a user's data, the expected gain from this would be negligible since the computation time, $T_2$, is short compared to $T_1$ and $T_3$ \cite[Figure~6]{Malik2020} and further optimizing each user's computation time at the MEC can significantly increase the formulation and algorithm complexity.

For our formulation to follow in Section \ref{formulation}, we consider equal frequency allocation for users' tasks, that is $f_{m,i} = f_m \ \forall i$, based on previous results in \cite{Malik2020} showing that, in typical network settings, wireless transmission energy consumption is significantly dominant compared to the computation energy consumption and therefore dynamic frequency allocation has little effect on the overall system's energy consumption.

\section{Optimization Problem Formulations}\label{formulation}
Considering a multi-cell multi-MEC network, we formulate an edge computing problem which explicitly accounts for physical layer parameters including available transmit powers from each user and the MEC, associated massive MIMO data rates with realistic pilot contamination and interference. For simplicity of notation, we assume that all $K$ users which are offloading their computation to the MEC server are also requesting wireless charging. 

In this section, we discuss a sequential formulation and consider the problems of computation offloading $(P_{\text{CO}})$ and wireless charging $(P_{\text{WC}})$ independently in terms of energy optimization. The aim of $(P_{\text{CO}})$ is to minimize the energy consumption for computation offloading, while the goal of $(P_{\text{WC}})$ is to maximize the energy received at the users through wireless charging. We consider wireless charging as an opportunistic service in the sense that charging happens during the time available after timing has been optimally allocated for computation offloading. This leads to a sequential optimization process where the optimization for wireless charging will follow that of computation offloading. It should be emphasized that the sequential process is only in terms of optimization, as once all the variables and parameters are optimized, the operations of computation offloading and wireless charging can occur simultaneously as discussed in the system model of Section \ref{sys_model}.

\subsection{Minimization Of Energy Consumption For Computation Offloading}
Using the uplink and downlink transmission rates, respectively defined as $r_{u,i} = \frac{s_i}{\nu t_{u,i} B}$ and $r_{d,i} = \frac{\mu s_i}{t_{d,i} B}$, and based on (\ref{rate_ul}) and (\ref{rate_dl}), we can express the per-user power allocation variables for uplink ($p_{li}$) and downlink ($\eta_{li}$) transmissions as functions of the time allocation and data partitioning as follows:
\begin{align}\label{poweralloc}
p_{li} = \frac{(2^{\frac{s_i}{\nu t_{u,i} B}} - 1)\Gamma_{1}\sigma_{1,i}^2}{N \gamma_{i}}, \  \ \eta_{li} = \frac{(2^{ \frac{\mu s_i}{t_{d,i} B}} - 1)\Gamma_{2}\sigma_{2,i}^2}{P N \gamma_{i}}
\end{align}
Replacing these expressions into (\ref{E_ul}) and (\ref{t_Li}), the total energy consumption by all users can be written as
\begin{equation}\label{E_u}
E_u =  \sum_{i=1}^{K} \left[\frac{t_{u,i}(2^{\frac{s_i}{\nu t_{u,i}B}} - 1)\Gamma_{1}\sigma_{1,i}^2}{N \gamma_i} + \kappa_i c_i (u_i - s_i) f_{u,i}^2 \right]
\end{equation}
Similarly, based on equations (\ref{E_dl}) and (\ref{tMEC}), the total energy consumption at the MEC server for computation offloading is
\begin{equation}\label{E_m}
E_m = \sum_{i=1}^{K} \left[\frac{t_{d,i}(2^{\frac{\mu s_i}{t_{d,i}B}} - 1)\Gamma_{2}\sigma_{2,i}^2}{N \gamma_{i}}  + \kappa_m d_m f_{mi}^2 s_i \right]
\end{equation}
The energy minimization problem for computation offloading can then be given as
\setcounter{equation}{13}
\begin{align*}\label{PseqCO}
(P_\text{CO}): \ \min_{\boldsymbol{s, t}} \ &E_{\text{total}} = (1 - w) E_{u} + w E_{m}  \tag{\theequation}&\\
\text{ s.t. }  \ &\text{Eqs. } (\ref{E_u})-(\ref{E_m}) \tag{a-b}&\\
&\sum_{j=1}^{3} \left ( T_j\right ) \leq T_d, \ \ \ \ \frac{c_i (u_i - s_i)}{f_{u,i}} + t_{u,i} - T_d \leq 0 \ \ \ \  \ \ \ \  \forall i \in [1,K] \tag{c-d}&\\
&t_{u,i} - T_1 \leq 0, \ \ \ \  \ t_{d,i} - T_3 \leq 0, \ \ \ \ \frac{d_m s_i}{f_{mi}} - T_2 \leq 0 \ \ \ \forall i \in [1,K] \tag{e-g}&
\end{align*}
Here $E_{\text{total}}$ is weighted sum of energy consumed at all users ($E_{u}$) and the MEC ($E_{m}$), with $1 - w$ and $w$ as the respective weights. The optimizing variables of this problems are time allocation $\boldsymbol{t} = [t_{u,1}...t_{u,K}, t_{d,1}...t_{d,K}, T_1, T_2, T_3, T_c]$, and offloaded data $\boldsymbol{s} = [s_1...s_K]$. Given parameters of the problems are $T_d$ as the total latency constraint, $P$ as the AP's transmit power, $B$ as the channel bandwidth, $\Gamma_1$, $\Gamma_2$ as the uplink and downlink capacity gaps, $(\kappa_i, c_i)$ and $(\kappa_m, d_m)$ as the switched capacitance and CPU cycle information at the users and the MEC respectively.

Constraints (a-b) show the total energy consumption at the users and the MEC respectively, which includes the energy consumed for offloading/downloading and computation. Constraints (c-d) represent the constraint that both the time consumed for all three phases at the MEC, and the time consumed for offloading $\boldsymbol{t_u}$ and local computation at each user $\boldsymbol{t_L}$ should not exceed $T_d$. Constraints (e-g) show that the time consumed separately for offloading $\boldsymbol{t_u}$, computation of users' tasks at the MEC $\boldsymbol{t_M}$, and downloading time $\boldsymbol{t_d}$ for each user's results must be less than the maximum allowable time, $\{ T_1, T_2, T_3\} $, for that phase as given in \{(\ref{E_ul}),(\ref{tMEC}), (\ref{E_dl})\} respectively. 

\subsection{Maximization Of Received Energy By Wireless Charging}
The above computation offloading problem is followed by the opportunistic wireless charging problem as given below
\setcounter{equation}{15}
\begin{align*}\label{PseqWC}
(P_{\text{WC}}): \ \max_{\boldsymbol{W_q}} \ \ &\sum_{i=1}^{K} \xi_i \text{tr}(h_i^\ast \boldsymbol{W_q} h_i) T_c \tag{\theequation}\\
\text{s.t.} \ \ &\text{tr}(\boldsymbol{W_q}) \leq P \tag{a}\\
&\xi_i \text{tr}(h_i^\ast \boldsymbol{W_q} h_i) T_c \leq e_i \ \ \forall i = 1...K \tag{b}
\end{align*}
Here the charging time is defined as $T_c = T_d - T_1^\star - T_3^\star$, where $T_1^\star$ and $T_3^\star$ are the optimal time allocation for offloading and downloading operations obtained by solving $(P_{\text{CO}})$. In this way, the two problems are formulated in a sequential manner in compliance with the overall latency constraint. The charging time $T_c$ denotes that wireless charging occupies all the time within $T_d$ outside the data transmission operations of offloading and downloading. The optimizing variable is the beamforming matrix for wireless charging $\boldsymbol{W_q} \in \mathbb{R}^{N \times N}$. The objective function is a sum of the received energy for all users and the objective is to maximize this overall received energy at the users. Constraint (a) represents the physical layer constraint on the maximum transmission power of the AP. Constraint (b) shows that the amount of received (charged) energy at the $i^{\text{th}}$ user is no more than the energy that it requests.

Problem $(P_{\text{CO}})$ is a semi-definite programming problem where the objective function and constraints are linear trace functions of $\boldsymbol{W_q}$ and hence convex. We can show that strong duality holds since Slater's condition is satisfied, that is, we can find a strictly feasible point ($\boldsymbol{W_q} = p\boldsymbol{I}_{N \times N}$, $p \leq P/N$) in the relative interior of the domain of the problem where the inequality constraints hold with strict inequalities~\cite{Boyd2004}. 

\section{Data Partitioning And Time Allocation for Computation Offloading}
\subsection{Problem Analysis}
In this section we analyze the computation offloading problem $(P_{\text{CO}})$ and show that it can be decomposed into simpler problems. The multivariable problem in (\ref{PseqCO}) is a non-linear and non-convex optimization problem. Following a similar approach as in \cite{Malik2020}, the objective function $f_0$ for $(P_{\text{CO}})$ is a convex function of $s_i$. Furthermore, provided that the gradient of $f_0(\cdot)$ with respect to $s_i$ evaluated at $s_i = 0$ is positive, which is often satisfied in typical network settings, then the total energy in problem $(P_{\text{CO}})$ is an increasing function of each $s_i$ and there exists an optimal point, $s_i^\star \ \forall i \in [1,K]$, which minimizes $E_{\text{total}}$ within the latency constraint. If offloaded data $\boldsymbol{s}$ is fixed, then problem $(P_{\text{CO}})$ turns out to be convex in the remaining variables as stated in the following lemma. Lemma \ref{lemma1} lets us decompose the original non-convex problem $(P_{\text{CO}})$ into simpler convex subproblems which will be used in the subsequent algorithm design. 
\begin{lemma}\label{lemma1}
For a given set of offloaded data $\boldsymbol{s}$, the problem $(P_{\text{CO}})$ is convex in the time allocation variable $\boldsymbol{t}$.
\end{lemma}
\begin{proof}
Proof follows by examining each constraint and showing that with fixed $s_i$, it is a convex function. Details in Appendix A.
\end{proof}

Since CPU frequencies are not optimizing variables, for given $s_i$ in $(P_{\text{CO}})$, we can find in closed form the optimum time consumed by the MEC to compute each user's tasks, and the overall time $T_2$ spent for the data computation function at the MEC as in Lemma \ref{prop1} next.

\begin{lemma}\label{prop1}
For a given value of the offloaded data $s_i$, the computation time for the offloaded data $T_2$ can be pre-determined in closed form as follows
\begin{equation}\label{T2closed}
T_2 = \max_i \frac{d_m s_i}{f_{mi}}
\end{equation}
and hence constraint (\ref{PseqCO}g) can be excluded from the problem $(P_{\text{CO}})$.
\end{lemma}
\begin{proof}
Directly from constraint (g) in (\ref{PseqCO}) for a given $s_i$.
\end{proof}

\subsection{Optimal Primal Solution}
Next we present the solution for the optimal time allocation for the computation offloading problem $(P_{\text{CO}})$. Since the problem is convex based on Lemma \ref{lemma1}, we adopt a primal-dual solution using the Lagrangian duality analysis similar to that proposed in \cite{Malik2020} and derive the optimal solution as given in Theorem~\ref{theorem2} below.

\begin{theorem}\label{theorem2}
The offloading and downloading time, $t_{u,i}$ and $t_{d,i}$ respectively, can be obtained as a solution of the form 
\begin{equation}\label{LambertSol}
x = \frac{c B}{\ln 2} \Big(W_0 \Big(\frac{-y}{\sigma^2 e} - \frac{1}{e}\Big) + 1 \Big)
\end{equation}
where $y = -\frac{\beta_i + \theta_i}{(1 - w)}$, $x = x_{1,i} = \frac{1}{t_{u,i}}$, $c = \frac{\nu}{s_i}$, $\sigma^2 = \frac{ \Gamma_{1} \sigma_{1,i}^2}{N \gamma_i}$ to solve for $t_{u,i}$, and  $y = \frac{- \phi_i}{w}$, $x = x_{2,i} = \frac{1}{t_{d,i}}$, $c = 1/\mu s_i$, and  $\sigma^2 = \frac{\Gamma_{2} \sigma_{2,i}^2}{N \gamma_i}$ to solve for $t_d,i$. Here $\theta_i$, $\beta_i$ and $\phi_i$ are the dual variables associated with the constraints (d), (e) and (g) of problem $(P_{\text{CO}})$ in (\ref{PseqCO}) respectively.
\end{theorem}
\begin{proof}
The solution in (\ref{LambertSol}) can be obtained directly by applying KKT conditions on the Lagrangian dual of the problem $P_{\text{CO}}$ with respect to $t_{u,i}$ and $t_{d,i}$. Detailed proof can be obtained using an approach similar to that in \cite[Theorem 1]{Malik2020} and is omitted for brevity.
\end{proof}

\section{Energy Beamforming for Wireless Charging}
In this section, we derive the solution for the optimal transmit covariance matrix, $\boldsymbol{W_q}$ by finding the optimal energy beam directions and also the optimal beam power allocation. For the received energy maximization problem $(P_{\text{WC}})$, we use Lagrangian duality analysis to obtain the optimal beam directions as described in Theorem \ref{theorem5} below.
\begin{theorem}\label{theorem5}
For maximizing the received energy, the optimal directions for energy beams are $\boldsymbol{U_q^\star} = \boldsymbol{U_C}$, where $\boldsymbol{U_C}$ is obtained from the eigenvalue decomposition of $\boldsymbol{C} = \boldsymbol{U_C \Lambda_C U_C^\ast}$, such that $\lambda_{C,1} \geq \lambda_{C,2} \geq \ldots \geq \lambda_{C,N}$, where
\vspace{-2mm}
\begin{equation}\label{matC}
\boldsymbol{C} =  \chi \boldsymbol{I} +  \xi_i T_c \sum_{i=1}^{K} (1 + \rho_i) \boldsymbol{h_i h_i^\ast}
\end{equation}
Here $\chi$ and $\rho_i$ are the dual variables associated with constraint (\ref{PseqWC}a) and the $i^{\text{th}}$ constraint in (\ref{PseqWC}b) respectively.
\end{theorem}
\begin{proof}
See Appendix B.
\end{proof}

Theorem \ref{theorem5} provides the optimal directions of the energy beams for the beamforming matrix, $\boldsymbol{W_q}$. What is left now is to obtain the optimal power allocation across the energy beams, that is, the eigenvalues of the transmit covariance matrix for wireless charging. To this end, we substitute the optimal beam directions from Theorem \ref{theorem5} into $(P_{\text{WC}})$ and re-write the formulation in terms of the beam power allocation only as ($P_{\text{BP}}$) below. Beam power allocation, $\boldsymbol{\lambda_q}$, can then be obtained as a solution to a Linear Programming (LP) problem given in Theorem \ref{theorem4} below.

\begin{theorem}\label{theorem4}
The optimal beam power allocation which maximizes the received energy through wireless charging is derived as a solution of the LP problem below
\setcounter{equation}{19}
\begin{align*}\label{P4}
(P_{\text{BP}}): \ \max_{\boldsymbol{\lambda_q}} \ \ &\sum_{i=1}^{K} \boldsymbol{d_i^\ast \lambda_q} \tag{\theequation}\\
\text{s.t.} \ \ & \sum_{i=1}^K \lambda_{q,i}  \leq P, \ \ \ \  \lambda_{q,1} \geq ... \geq \lambda_{q,K} \geq 0 \tag{a-b}\\
&\boldsymbol{D \lambda_q} \leq \boldsymbol{b} \tag{c}
\end{align*}
where $\boldsymbol{\lambda_q} = [\lambda_{q,1}, ..., \lambda_{q,K}]^T$, $\boldsymbol{D} \in \mathbb{R}^{K \times K} = [\boldsymbol{d_1^\ast}...\boldsymbol{d_K^\ast}]$, $\boldsymbol{d_i}^\ast = \text{\emph{diag}}(\boldsymbol{r_i r_i^\ast})$, $\boldsymbol{r}_i^\ast = \boldsymbol{h}_i^\ast  \boldsymbol{U_C} = \boldsymbol{h}_i^\ast  \boldsymbol{U_q}^\star$ and $\boldsymbol{b} \in \mathbb{R}^{K \times 1} = [\pi_1 ... \pi_K]$, $\pi_i = \frac{e_i}{\xi_i T_c} \ \forall i = 1...K$. 
\end{theorem}
\begin{proof}
Obtained by substituting optimal beam directions from Theorem \ref{theorem5} in $(P_{\text{WC}})$. Details in Appendix C.
\end{proof}

Note that in the above solutions for $(P_{\text{WC}})$, since the goal is energy maximization, the eigenvalues of $\boldsymbol{W_q}$ and $\boldsymbol{C}$ are of the same order. All the eigenvectors of each matrix are ordered according to their corresponding eigenvalues. The optimal solutions derived thus far are specific to the respective problems $(P_{\text{CO}})$ and $(P_{\text{WC}})$, and thus reveal the optimal solution structure that otherwise would be obscured by using a generic solver. Next we use these optimal solutions to design customized algorithms to solve these problems.

\section{Algorithm Design}
\subsection{Sequential And Nested Algorithm Structures}
\setlength{\belowcaptionskip}{-20pt}
\begin{figure}[t]
\centering
\includegraphics[scale = 0.6]{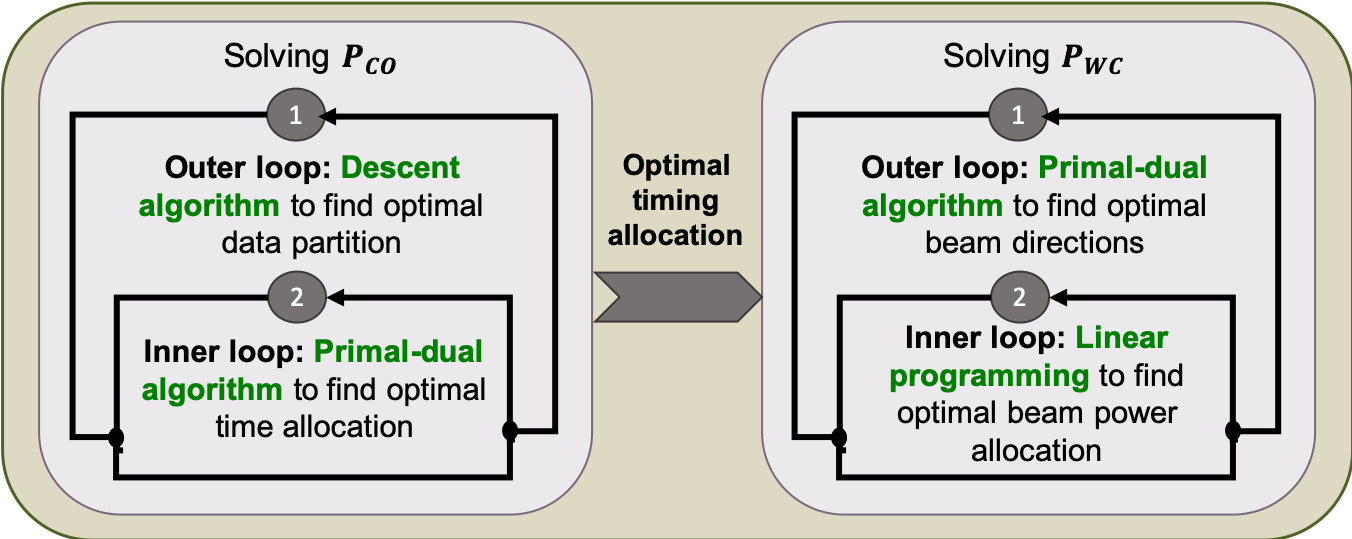}
\caption{Algorithm architecture for $P_{\text{CO}}$ and $P_{\text{WC}}$}
\label{solution_flow}
\end{figure}
In this section, we discuss the algorithm structure to solve the two sequentially formulated problems $P_{\text{CO}}$ and $P_{\text{WC}}$. Based on the way these two problems are formulated, $P_{\text{CO}}$ will be solved first to obtain the optimal data partitioning and time allocation for computation offloading. This optimal time allocation will then be used in $P_{\text{WC}}$ as a given parameter in order to find the optimal energy beamforming structure.

The algorithm for solving $(P_{\text{CO}})$ is designed based on Lemma \ref{lemma1} to have a nested architecture with an outer and an inner loop, in which the outer loop solves for $s_i$ decrementally while the inner loop solves for $\boldsymbol{t}$ at a fixed value of $s_i$. Specifically, the nested algorithm works as follows. We first initialize the offloaded bits $\boldsymbol{s}$ and the dual variables in the outer algorithm. At the current value of $s$, the inner algorithm is executed, for which we use a primal-dual approach employing a subgradient method. At convergence where the stopping criterion for the dual problem is satisfied, the inner algorithm returns the control to the outer algorithm. Based on the newly updated primal solution from the inner algorithm, we proceed to updating $\boldsymbol{s}$ by some $\Delta s_i$ for each user for the next iteration of the outer algorithm, using a latency aware descent algorithm. Similar to \cite{Malik2020}, the latency aware descent algorithm is based on the standard Newton method, where the computation cost for each Newton iteration requires $\mathcal{O}(n^3)$ flops~\cite{Malik2020}, with a novel modification to the classical stopping criterion to account for the latency constraint.
\begin{algorithm}[t]
\caption{Solution for $(P_{\text{CO}})$ and $(P_{\text{WC}})$} 
\text{Given:} Distances $d_{i} \ \forall i$. Channel $\boldsymbol{H = G^{T}}$. Precision, $\epsilon_1, \epsilon_2$, Data $u_i$, Latency $T_d$. \text{Initialize:} $s_i$\\
\textbf{Begin Outer Algorithm for $(P_{\text{CO}})$}\\
\textbf{Given} a starting point $\boldsymbol{s}$, \textbf{Repeat}
\begin{enumerate}
\item Compute $\Delta \boldsymbol{s}$
\item \textbf{Begin Inner Algorithm}
\begin{itemize}[leftmargin=*]
\item Calculate $t_{u,i}$ and $t_{d,i}$, using (\ref{LambertSol}). Then $T_1^\star = \max t_{u,i}^\star$ and $T_3^\star = \max t_{d,i}^\star$.
\item Update $p_i$ and $\eta_i$ using (\ref{poweralloc}) and calculate $\sigma_{1,i}^2$ and $\sigma_{2,i}^2$.
\item Find dual function in (\ref{DualFn}), stop if dual variables converge with $\epsilon_2$, else find subgradients in (\ref{subgrads}-d), update dual-variables using subgradient method and continue
\end{itemize}
\textbf{End Inner Algorithm}
\item \textit{Line search and Update}. $s_i := s_i + t_i\Delta s_i$.
\end{enumerate}
\textbf{Until} stopping criterion is satisfied with $\epsilon_1$ or latency constraint $T_d$ is met.\\
\textbf{End Outer Algorithm for $(P_{\text{CO}})$}\\
\text{Given:} Optimal time allocation from (P2) in Step 2 above, find $T_c = T_d - T_1^\star - T_3^\star$\\
\textbf{Begin Algorithm for $(P_{\text{WC}})$}
\begin{itemize}
\item Find  $\boldsymbol{\lambda_q}^\star$ from ($P_{\text{BP}}$), and $\boldsymbol{W_q}^\star = \boldsymbol{U_C \Lambda_q^\star U_C^\ast}$, where $\boldsymbol{\Lambda_q^\star} = \textbf{diag}(\boldsymbol{\lambda_q}^\star)$
\item Find dual function in (\ref{PWCdual}), Stop if dual variables converge with $\epsilon_2$, else update using subgradients in (\ref{subgrads})
\end{itemize}
\textbf{End Algorithm for $(P_{\text{WC}})$}
\end{algorithm}
Problem $(P_{\text{WC}})$ solves for the transmit covariance matrix $\boldsymbol{W_q}$ as an independent problem after obtaining the optimal time allocation solution from $(P_{\text{CO}})$ to calculate the charging time $T_c$ as $T_c = T_d - T_1^\star - T_3^\star$. The algorithm for solving $(P_{\text{WC}})$ also has a nested structure, with an outer algorithm to establish the optimal beam directions and an inner algorithm for the beam power allocation. Specifically, at each iteration of $(P_{\text{WC}})$, an outer algorithm step finds the optimal dual variables for the beam direction solutions in Theorem \ref{theorem5} via a subgradient method, and calls to an inner algorithm which solves the LP problem ($P_{\text{BP}}$) in Theorem \ref{theorem4} for the optimal beam power allocation using a standard convex solver. Once the beam power allocation is found, the inner algorithm returns to the outer one in order to update the dual variables, and the process continues until convergence is reached in the outer algorithm. In the case of $(P_{\text{WC}})$, the outer algorithm is primal-dual, and the inner algorithm is linear programming. The algorithm flow is depicted in Figure \ref{solution_flow} and steps for solving both problems $(P_{\text{CO}})$ and $(P_{\text{WC}})$ are given in Algorithm 1.

\subsection{Primal-Dual Algorithms}
For the inner optimization in $(P_{\text{CO}})$ and the outer algorithm in $(P_{\text{WC}})$, we design primal-dual algorithms where the primal variable are obtained as closed-form functions of the dual variables, which are found by solving the dual problem using a sub-gradient method. The dual-function for the convex optimization problem $(P_{\text{CO}})$ at a given $s_i$ can be defined as
\begin{equation}\label{DualFn}
    g_{\text{CO}}(\lambda_1,\boldsymbol{\beta, \xi_i, \phi}) = \inf_{\boldsymbol{t}} \mathcal{L}_{\text{CO}}(\boldsymbol{t}, \lambda_1, \boldsymbol{\beta, \xi_i, \phi})
\end{equation}
where $\mathcal{L}_{CO}$ is the Lagrangian for problem $(P_{\text{CO}})$  and the dual-problem is defined as
\begin{align}\label{PDual}
    P_{\text{CO}}\text{-dual: }\max \ &g_{\text{CO}}(\lambda_1,\boldsymbol{\beta, \xi_i, \phi})  \ \  \text{s.t. } \lambda_1 \geq 0, \beta_i, \theta_i, \phi_i \geq 0 \ \forall i = 1...K 
\end{align}
where $\lambda_1$, $\boldsymbol{\beta, \xi_i}$, and $\boldsymbol{\phi}$ are the dual variables associated with constraints (c-f) in (\ref{PseqCO}), respectively.

Based on the dual-function for the problem $(P_{\text{WC}})$ in (\ref{PWCdual}), the dual problem is given as
\begin{align}\label{Pseqdual}
 \text{$\text{P}_{\text{WC}}$-dual: }\: \min \ \ &g_{\text{WC}}(\boldsymbol{\rho}, \chi) \ \  \text{s.t.} \ \ \chi \geq 0, \rho_i \geq 0 \ \text{for } i = 1...K 
 \end{align} 
Using the closed form expressions for the primal variables in terms of the dual-variables as in Theorems \ref{theorem2}-\ref{theorem5}, the dual functions above are functions of only the dual-variables. 

The subgradient terms with respect to all dual variables of original problems $(P_{\text{CO}})$ and $(P_{\text{WC}})$ are as given below
\squeezeup
\setcounter{equation}{23}
\begin{align*}\label{subgrads}
&\nabla_{\lambda_1}\mathcal{L} = \sum_{j=1}^{3} T_j - T_{\text{delay}} \tag{\theequation a}\\
&\nabla_{\beta_i}\mathcal{L} = t_{u,i} - T_1, \ \ \nabla_{\phi_i}\mathcal{L} = t_{d,i} - T_3, \ \ \nabla_{\theta_i}\mathcal{L} = \frac{c_i q_i}{f_{u,i}}  + t_{u,i} - T_d, \ \ \ \ \ \  \tag{b-d}\\
&\nabla_{\rho_i}\mathcal{L} = \xi_i \text{tr} \left( \boldsymbol{h_i^\ast W_q h_i} \right )T_c  - e_i,  \ \ \ \ \  \ \nabla_{\chi}\mathcal{L} = \text{tr}(\boldsymbol{W_q}) - P \tag{e-f}
\end{align*}

For implementation of the primal-dual algorithms, we use the subgradient method to solve the constrained convex optimization problems $(P_{\text{CO}})$ and $(P_{\text{WC}})$~\cite{Boyd2003}. The designed algorithms find the subgradients for the negative dual function $-g_{\text{CO}}$, since the dual problem in (\ref{PDual}) is a maximization problem for the dual function, and for the positive dual function $g_{\text{WC}}$, since the dual problem in (\ref{Pseqdual}) is a minimization problem. At each iteration, the primal variables are updated based on Theorems \ref{theorem2}-\ref{theorem4}. The dual variables vector $x$ is updated as $\boldsymbol{x}^{(k+1)} = \boldsymbol{x}^{(k)} - \beta_k \boldsymbol{g}^{(k)}$, where $\beta_k$ is the $k^{\text{th}}$ step-size, and $\boldsymbol{g}^{(k)}$ is the subgradient vector at the $k^{\text{th}}$ iteration evaluated using the sub-gradient expressions in (\ref{subgrads}-f). We use the non-summable diminishing step size, setting $\beta_k = 1/\sqrt{k}$, using which the algorithm is guaranteed to converge to the optimal value with a theoretical iteration complexity of $\mathcal{O}(1/\epsilon^2)$~\cite{Boyd2003}\cite{Ryan2015}. Since the subgradient method is not a descent method, the algorithms keep track of the best point for the dual functions at each iteration of the inner algorithm. These primal-dual update steps are repeated until the desired level of precision, $\epsilon_2$, is reached for the stopping criterion.

For the diminishing step size as that considered, the subgradient method is guaranteed to converge as $k \to \infty$~\cite{Boyd2003}. In the subgradient method, since the key quantity is not the function value but rather the Euclidean distance to the optimal set~\cite{Boyd2003}, therefore, for our implementation we define the stopping criterion as: $\lVert \boldsymbol{g}^{(k+1)} - \boldsymbol{g}^{(k)} \rVert_2 \leq \epsilon_2$. 

\section{Numerical Results}
\setlength{\belowcaptionskip}{-20pt}
\begin{figure}[t]
                \centering
                \includegraphics[scale=0.375]{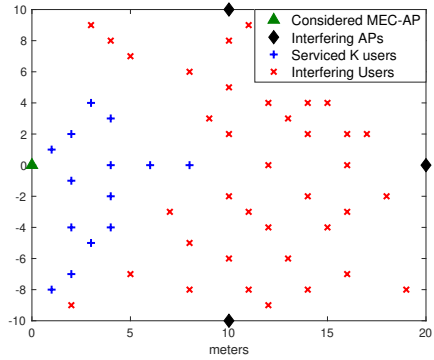}
                \caption{Network Layout}
                \label{system_model}
\end{figure}
In this section, we evaluate the solution of the sequential problem formulation with respect to energy and time consumption, the partition of bits offloaded to the MEC for computation and the received energy via wireless charging. We consider a $20\text{m}\times20\text{m}$ area (typical service area for AR applications with bi-directional transmission \cite{22.104}) with 4 APs and 16 users randomly located with $K = 4$ users per AP's coverage area and $N = 100$ as shown in Figure~\ref{system_model}. For simulations, $w = 10^{-3}$, $T_d = 20$ms (for AR/VR applications~\cite{Intel2017}), $B = 5$MHz, $\tau_c = B T_d$, $\Gamma_1 = \Gamma_2 = 1.25$, $\mu = 2$, $\kappa_i = 0.5$pF, $\kappa_m = 5$pF, $c_i = 1000$, $d_m = 500$, $\gamma = 2.2$, $\sigma = 2.7$dB, $\sigma_r^2 = -127$dBm, $\sigma_k^2 = -122$dBm, $f_{u,i} = f_u = 1800$ MHz $\forall i$. Each MEC processor has 24 cores with maximum frequency of $3.4$GHz, and we use $f_{m,i} = f_m = \frac{24 \times 3400}{K}$ MHz $\forall i$. Transmit power available at user and AP is 23 dBm and 46 dBm respectively. To calculate the interference and noise power ($\sigma_{1,i}^2$, $\sigma_{2,i}^2$) which include massive MIMO pilot contamination and intercell interference, we assume that user terminals transmit at their maximum power, that is $p_{qi} = 23$dBm, and the interfering APs use equal power allocation in the downlink, that is $\eta_{qi} = \frac{1}{K} \ \forall i$. Numerical results are averaged over 100 independent channel realizations of $\mathbf{H}$ and $\mathbf{G}$. The results with increasing number of users in the network are averaged over 200 spatial realizations (randomly generated user locations).

\subsection{Comparison of Wireless Charging Schemes}
\setlength{\belowcaptionskip}{-20pt}
\begin{figure}[t]
\centering
\includegraphics[scale = 0.6]{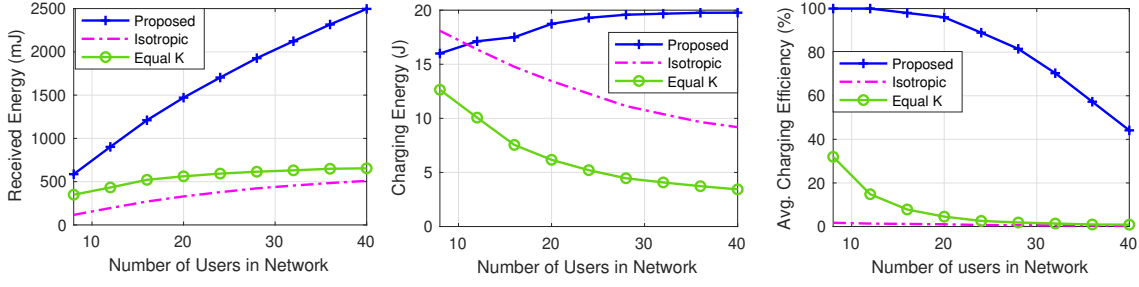}
\caption{Comparison of the proposed wireless charging scheme with isotropic wireless charging, and directed K-beam charging with equal power allocation (equal K)}
\label{beam_effect}
\end{figure}
Figure \ref{beam_effect} shows a comparison of the proposed maximization wireless charging scheme with two other schemes: (i) isotropic scheme where $\boldsymbol{W_q} = \frac{P}{N}\boldsymbol{I}$ and equal charging power $P/N$ is allocated across all $N$ antennas of the AP, and (ii) equal K with directional charging using the beamforming directions proposed in Theorem \ref{theorem5}, but with equal power allocation $P/K$ across $K$ energy beams. For fairness of comparison with the sequential scheme, we use power scaling for the other two schemes such that each user only receives an amount of energy at most equal to requested, similar to the sequential scheme. Since wireless charging is proposed as a billable service for future networks, this is also a necessary design consideration from the service providers' and consumers' perspectives. 

Figure \ref{beam_effect} shows the received energy on the left, the transmitted energy in the middle, and the average charging efficiency on the right. Average charging efficiency (per time block) is defined as the average percentage of received energy, in the $q^{th}$ time block as denoted by $(q)$, at the users end compared to the requested energy, given as
\begin{equation}
\text{Avg. Charging Efficiency }(\%) = \frac{{\sum_{i=1}^{K} \xi_i \text{tr}(h_i^\ast \boldsymbol{W_q} h_i) T_c}^{(q)}}{\sum_{i=1}^{K} e_i^{(q)}}
\end{equation}
Note that the requested energy at the $q^{th}$ time block excludes the amount of energy requests already fulfilled in the previous time block(s). As illustrated in this figure, the sum received energy for the energy maximization sequential scheme is significantly larger than the other two schemes. Beamforming with equal power allocation scheme performs better than the isotropic scheme, since it consumes lesser charging energy and still delivers higher energy to the users. Comparing the average charging efficiency for all the schemes, however, the opportunistic wireless charging maximization scheme enables substantially higher charging efficiency. The average efficiency is seen to decrease with an increase in the network size as expected.
\setlength{\belowcaptionskip}{-20pt}
\begin{figure}[t]
\centering
\includegraphics[scale = 0.65]{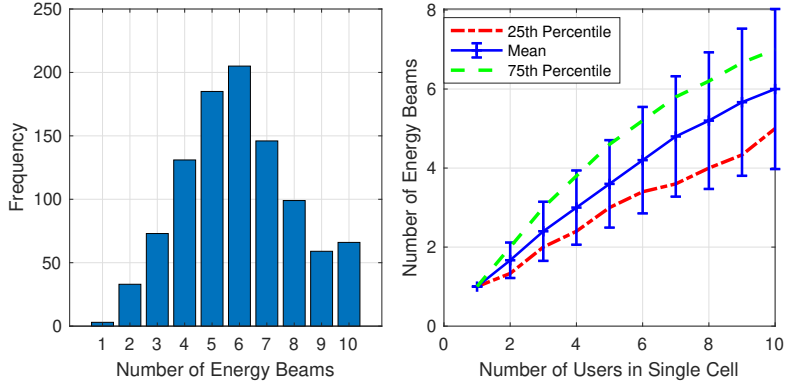}
\caption{Number of charging beams for K = 10 users (left), average number of charging beams for increasing no. of users in each cell (right).}
\label{num_beams}
\end{figure}
Another interesting finding presented in Figure \ref{num_beams} is the optimal number of energy beams for $K = 10$ users per cell and for an increasing number of users in the network. For the isotropic wireless charging, there are always $N > K$ energy beams. For the case of $K$ beams with equal power allocation, the number of beams is equal to the number of users in the cell. While multiple energy beams may be necessary for a multi-user system as also previously discussed in~\cite{Zeng2017}, the optimal number of energy beams for the proposed wireless charging scheme is usually less than the number of users. Since each energy beam can contribute as additional RF charging sources for neighboring users, the transmit beamforming can be intelligently designed as proposed to limit the number of energy beams which can prevent energy losses caused by transmitting energy in numerous directions. Therefore, for received energy maximization, the results show that on average, the optimal number of beam is much lower than the number of users in order to deliver the highest charged energy.

\subsection{Charging Profile}
\setlength{\belowcaptionskip}{-20pt}
\begin{figure}[t]
\centering
\includegraphics[scale = 0.5]{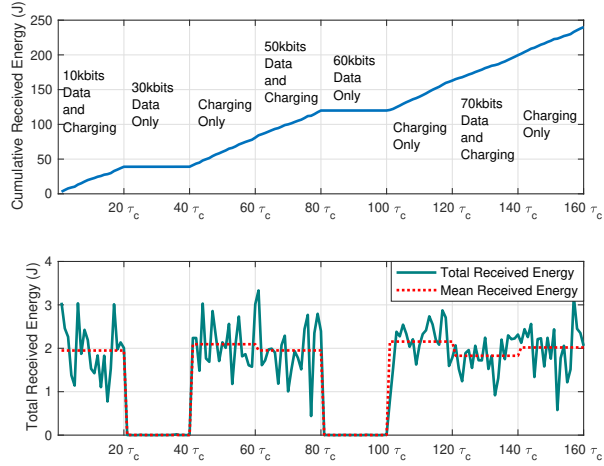}
\caption{A typical example of system performance for K users over time under three modes (a) Computation and Charging, (b) Computation only, and (c) Charging only}
\label{system_profile}
\end{figure}
\setlength{\belowcaptionskip}{-20pt}
\begin{figure}[t]
\centering
\includegraphics[scale = 0.5]{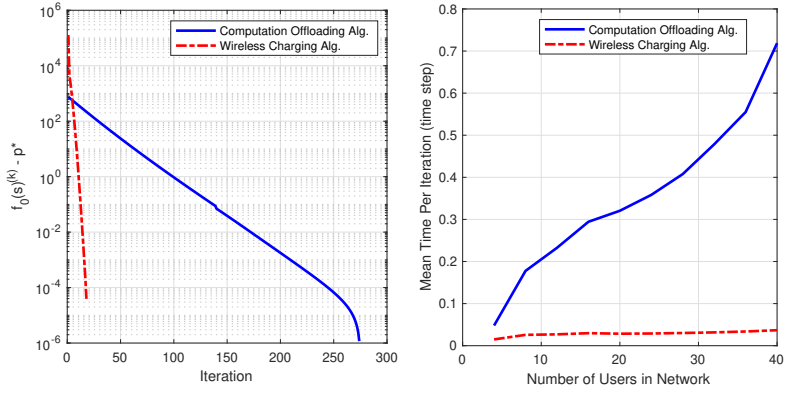}
\caption{Algorithm convergence (left) and mean time per iteration (right) for $P_{\text{CO}}$ and $P_{\text{WC}}$}
\label{alg_conv}
\end{figure}

Figure~\ref{system_profile} shows a typical example of the system's charging performance over time under the three modes of operation, namely, data and charging, data only, and charging only. The charging only and data only modes are special cases (or subsets) of the data and charging mode. For the joint data and charging mode, both data $(u_i)$ and energy requests $(e_i)$ are non-zero, that is $u_i>0 , e_i >0 \ \forall \ i$. For the data only mode, $e_i=0$ and for the charging only mode, $u_i=0$. The figure shows the time profile for the received energy by the users. The time axis is plotted in terms of the coherence interval $\tau_c$, to show that the energy values are calculated for a new channel realization after every coherence interval which corresponds to the variation in the received energy value over time in the bottom plot. For the results shown we assume $\tau_c = B T_d$, with $T_d = 20$ms. The cumulative energy on the top figure show that during the data only operation in which no users request wireless charging, there is no increase in the charged energy as expected. Correlating with the bottom plot, we see a decrease in the mean received energy during the joint phase of computation and power transfer (data and charging) as compared to the charging only phase. The results verify that our algorithm works as expected since with computation, a portion of time from $\tau_c$ is spent on data computation and wireless transmission, as compared to the charging only mode where the entire duration is spent for wireless charging. The cumulative top plot show that over an extended period of time over both computation and non-computation intervals, wireless charging can deliver a significant amount of energy.
\subsection{Algorithm Convergence}
Figure \ref{alg_conv} shows, on the left, the convergence of the two algorithms solving optimization sub-problems $P_{\text{CO}}$ and $P_{\text{WC}}$ with $u_i = u = 10\text{kbits}, \ e_i = e = 1\text{J} \ \forall i$. The algorithm for $P_{\text{WC}}$, or $P_{\text{WC}}$ \textit{algorithm} in short,  based on nested subgradient method and linear programming converges in significantly fewer iterations compared to the algorithm for $P_{\text{CO}}$, or $P_{\text{CO}}$ \textit{algorithm}, based on nested latency-aware Newton descent and subgradient methods. Not only does $P_{\text{WC}}$ \textit{algorithm} converge in fewer iterations compared to the $P_{\text{CO}}$ \textit{algorithm}, the time taken per iteration is also shorter for $P_{\text{WC}}$ as shown in Figure \ref{alg_conv} on the right. The computation offloading $P_{\text{CO}}$ \textit{algorithm} optimizes for data partitioning and time allocation for each user, leading to the number of optimizing variables for $K$ users as $2K$. Furthermore, these variables inherently affect the power allocation in (\ref{poweralloc}) and consequently the interference terms in uplink and downlink rate calculation in (\ref{rate_ul}) and (\ref{rate_dl}) for each user, this co-dependence adds complexity to the $P_{\text{CO}}$ \textit{algorithm} and requires a larger number of iterations to converge. These interdependent computations make the mean time per iteration increase almost exponentially with the number of users in the network as seen in the right plot of Figure \ref{alg_conv}. The wireless charging $P_{\text{WC}}$ \textit{algorithm}, on the other hand, calculates the beam directions for all $K$ users through a single matrix factorization per iteration as in Theorem 2. The power allocation per energy beam is then solved via an efficient inner linear programming algorithm which scales slowly with the number of users in the network. For our implementation on a personal computer, the time step unit in Figure \ref{alg_conv} is a second, however for faster machines, such as MEC servers, with the high-performance CPUs, this time-step may be significantly smaller.

\subsection{Effect of the Amount of Data Requested and Partitioning}
\setlength{\belowcaptionskip}{-20pt}
\begin{figure}[t]
        \begin{minipage}{0.5\textwidth}
                \centering
               \includegraphics[scale = 0.45]{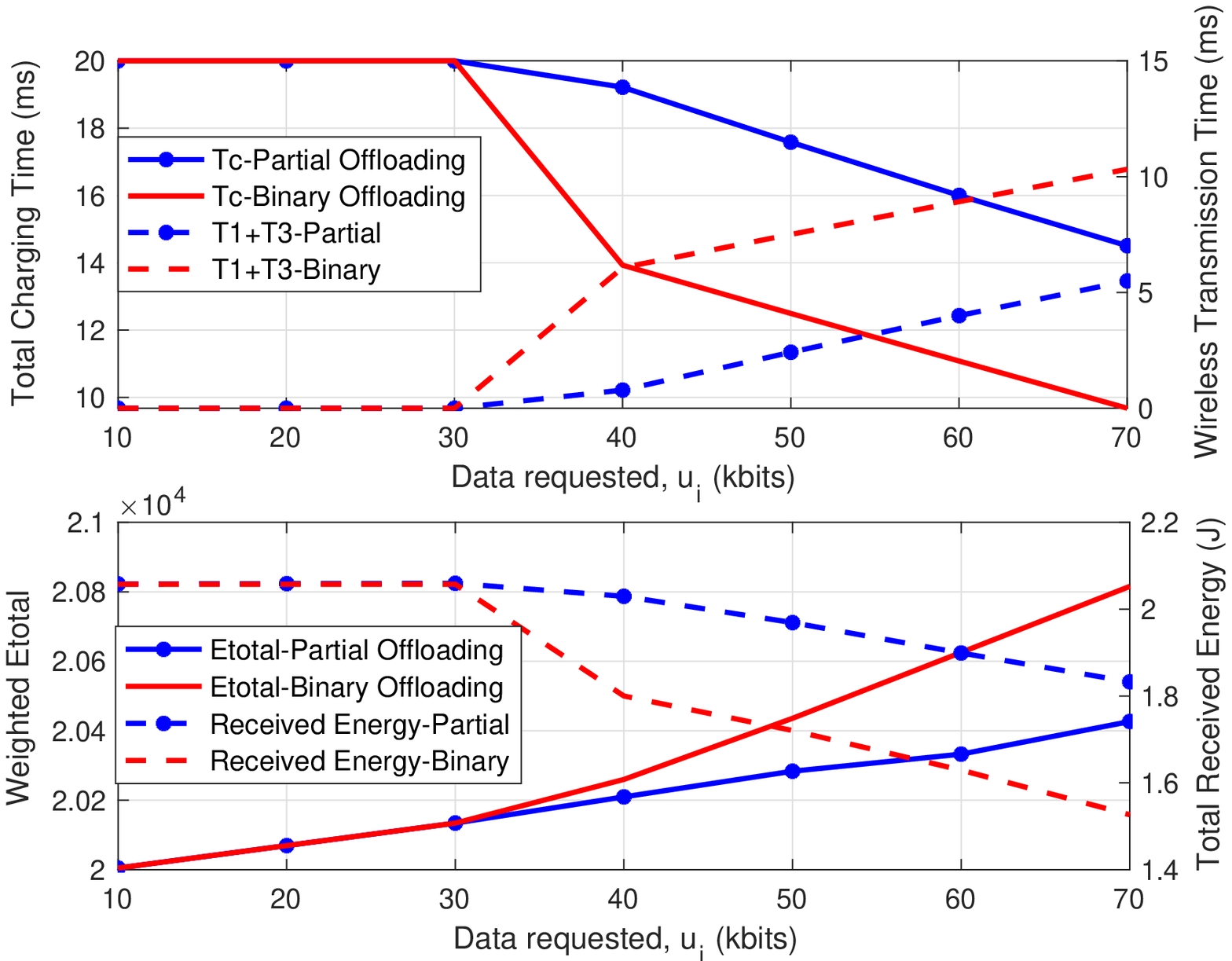}
                \caption{Time and energy consumption for partial and binary offloading schemes}
\label{data}
        \end{minipage}
       \hfill
        \begin{minipage}{0.5\textwidth}
                \centering
               \includegraphics[scale = 0.45]{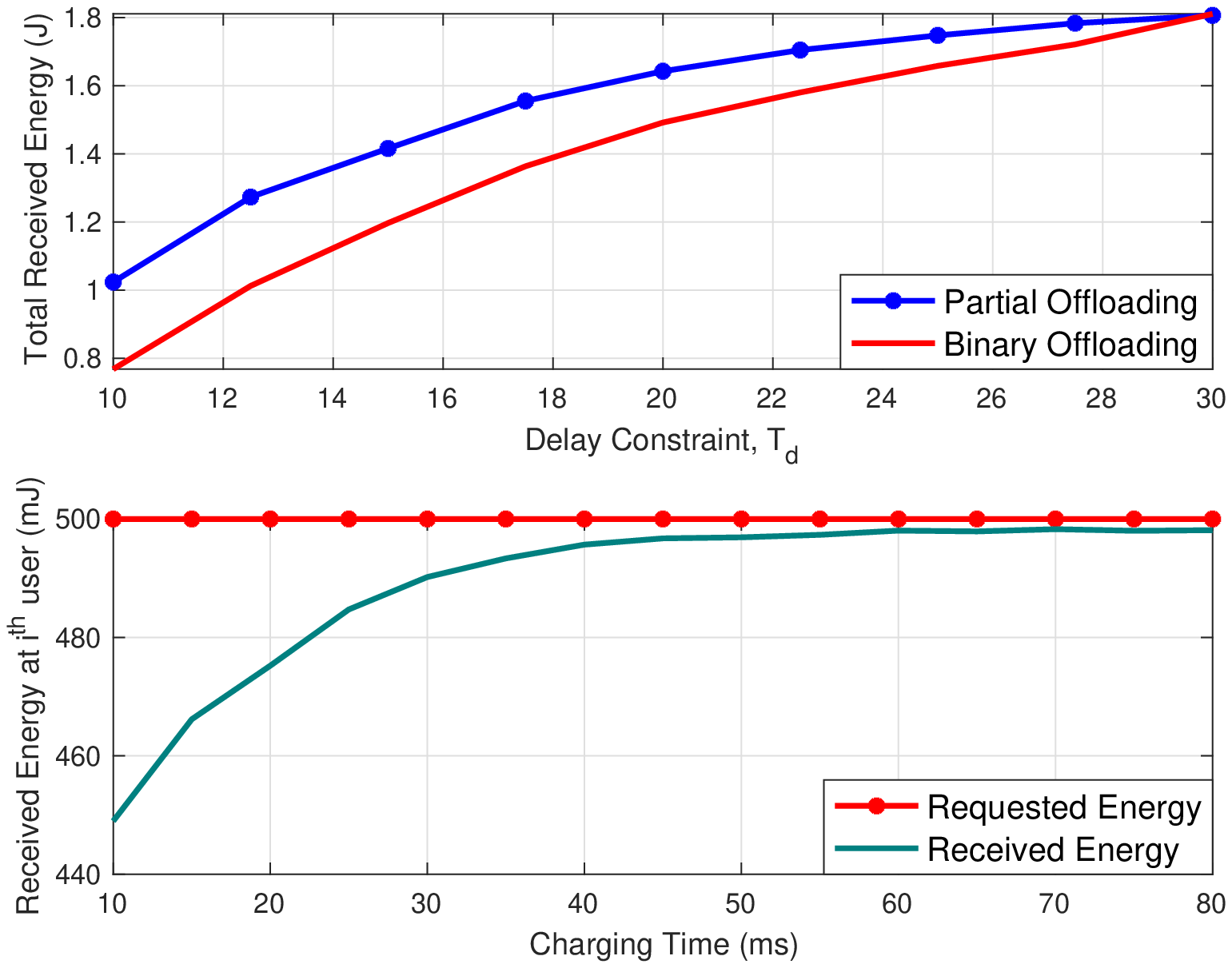}
\caption{Charged Energy Received for K = 4 users at (i) all users, each requesting 500mJ of energy (top), (ii) the $i^{\text{th}}$ user requesting 500mJ of energy (bottom) under partial offloading scheme}
\label{time_figure}
        \end{minipage}
\end{figure}

Figure~\ref{data} shows a comparison of the proposed partial offloading scheme, where data partitioning is used to divide the computation between the MEC and each user, with the binary offloading scheme where the task is atomic and is either offloaded or computed locally as a whole. We compare the time and energy consumption for the two schemes as the amount of data requested is increased under a fixed latency constraint of $T_d = 20$ms. To evaluate the solution for the binary offloading scheme, we consider all possible binary offloading combinations and choose the one with the lowest overall energy consumption. We see significant disparity between the binary and partial offloading schemes when large amounts of data are requested. For low data requests, local processing at users is optimal so both schemes consume the same energy and the entire duration is spent for wireless charging by the MEC concurrently with the local computation at the users. For larger data requests, however, the binary offloading scheme spends far less time for charging, since the time for wireless transmission to offload all the data to the MEC is greater. Owing to this increased time for wireless transmission in the binary scheme, the overall weighted system energy consumption is much larger than that of partial offloading. Partial offloading not only results in a lower overall weighted energy consumption, but also leads to higher received energy at the users  during wireless charging by the MEC because of the longer charging time. Partial offloading with data partitioning therefore appears as a potent design variable for the resource allocation problem, with significant impact on the wireless charging capability of the system.

\subsection{Effect of the Latency Constraint and Charging Time}
Figure \ref{time_figure} (top) shows the total received or charged at the users, each requesting 500mJ of energy, as the delay constraint is relaxed, that is, $T_d$ is increased at a fixed amount of requested data, $u_i = 50$ kbits. For this amount of data, binary offloading results in lower received energy since all data is offloaded to the MEC to meet the latency requirement. This results in larger time consumption for wireless transmission, consequently reducing the charging time and hence the charged energy. For relaxed latency, however, both binary and partial offloading schemes compute data locally, and hence the plots converge.

Figure \ref{time_figure} (bottom) shows the received energy, that is the amount of charge the user receives through wireless power transfer, as the charging time is increased. We show the requested and received energy for one user in a 16 user network, where each user requests 500 mJ of energy from the MEC, and the network is in \textit{charging only} mode, that is, the users do not request any data for computation. For longer charging times, the MEC fulfills the user's demand for wireless charging almost completely.

\subsection{Effect of the amount of Energy Requested}
Figure \ref{energy} shows the amount of energy received by a user through wireless charging, as the amount of energy requested by the user is increased in the \textit{data and charging} mode, that is the users jointly request data computation and wireless charging. We assume all the users requesting the same amount of energy, that is $e_i = e \ \forall i$. For lower amounts of requested energy, we see that the MEC-AP strives to fulfill the energy demand to a large extent, however, as the energy demands are increased by all the users simultaneously, the wireless charging by the MEC-AP cannot cope with the wireless charging demand in full. We compare the amount of energy received through wireless charging for three scenarios (a) all users request 70 kbits of data for computation, that is $u_i = 70$ kbits $\forall i$ under a latency requirement of $T_d = 20$ms, (b) all users request the same amount of data $u_i = 70$ kbits but the latency requirement is relaxed, that is $T_d = 40$ ms, and (c) all users request a smaller amount of data, that is $u_i = 30$ kbits for the strict latency requirement of $T_d = 20$ ms. When a lower amount of data is requested for computation, we see an increase in the received energy, since lesser time is spent for data transmission, which leaves more time for wireless charging. For a relaxed latency constraint, however, we see significant increase in the charged energy, since for a larger $T_d$ the time for power transfer is increased proportionally. 
\setlength{\belowcaptionskip}{-20pt}
\begin{figure}[t]
        \begin{minipage}{0.5\textwidth}
                \centering
                \includegraphics[scale = 0.45]{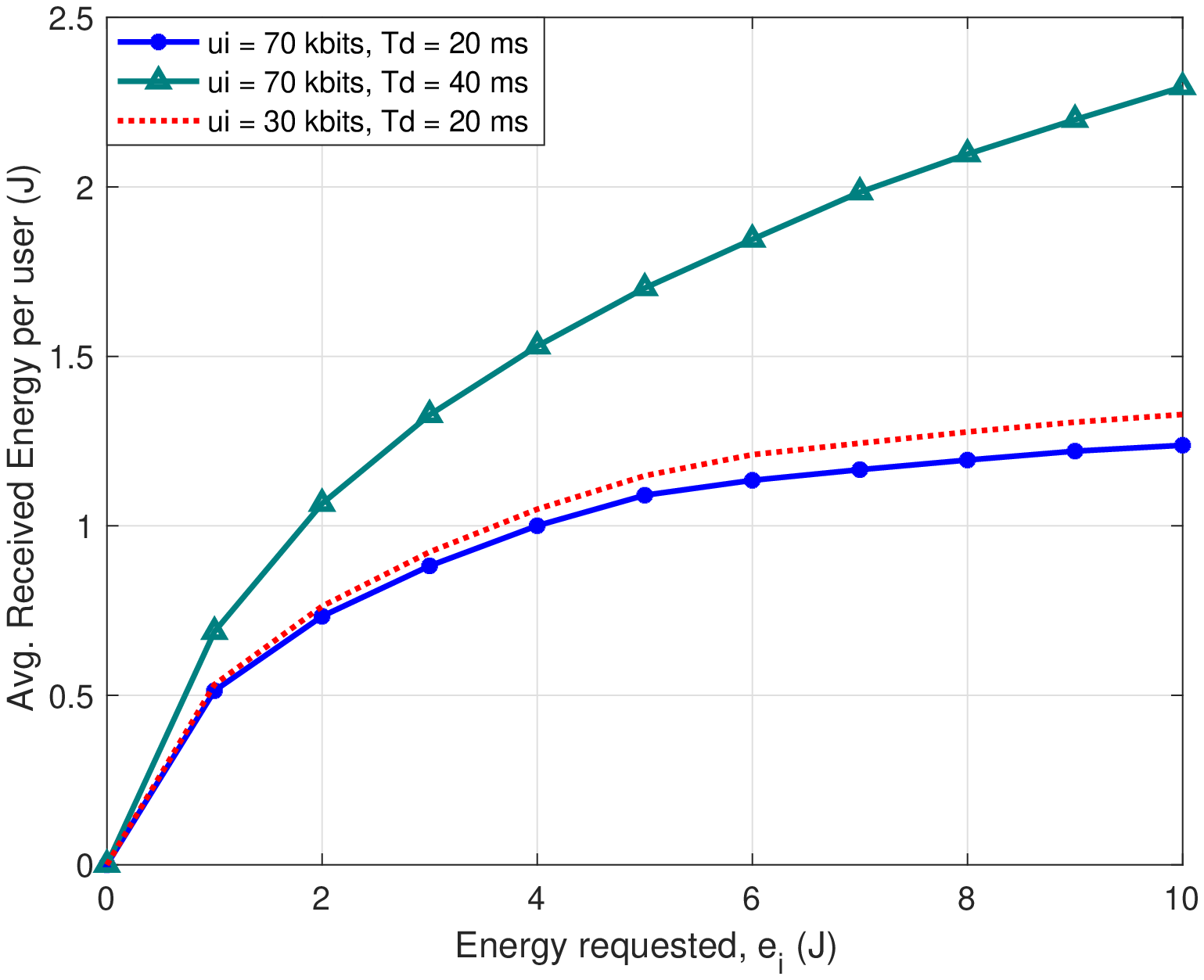}
\caption{Received energy by a user as it requests more energy under different computation and time requirements}
\label{energy}
        \end{minipage}
       \hfill
        \begin{minipage}{0.5\textwidth}
                \centering
                \includegraphics[scale = 0.45]{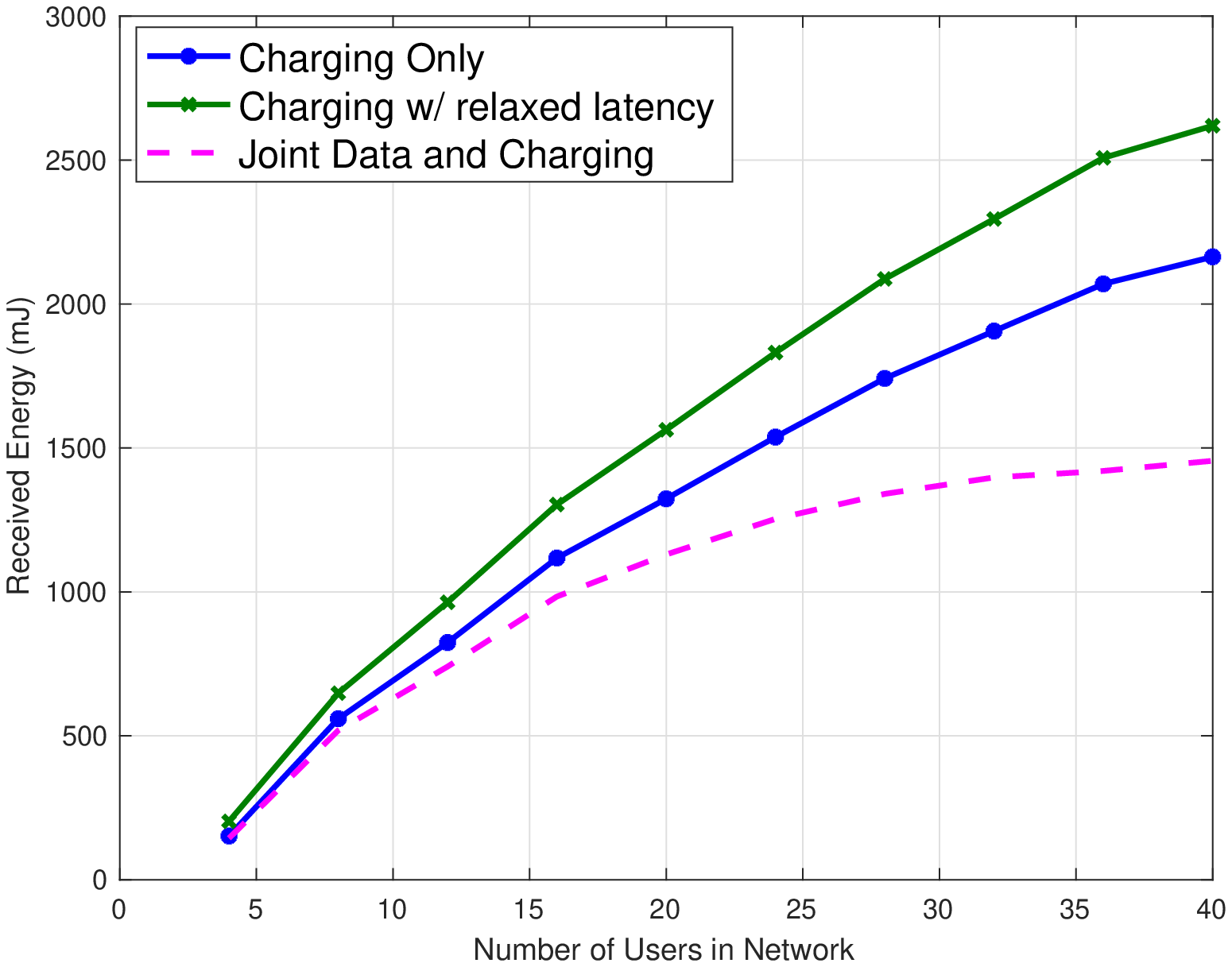}
\caption{Charged Energy Received at the Users Requesting 500mJ of Energy with and without data computation as the network is expanded}
\label{network_plot}
        \end{minipage}
\end{figure}
\subsection{Effect of Network Size}
Figure \ref{network_plot} shows the total amount of energy received by all users during the wireless charging function, as the number of users in the network is increased, each requesting $e_i = 500$mJ of energy. We compare the charged energy under different scenarios, namely (i) \textit{Charging Only} where each user only has energy requests and no data to offload, (ii) \textit{Charging w/ relaxed latency} with charging-only mode but the latency constraint is relaxed from 20 ms to $T_d = 40$ ms, and (iii) \textit{Joint Data and Charging} where each user request $u_i = 70$ kbits of data for computation along with its energy request. Relaxing the latency constraint, and having no data for computation, each option significantly increases the received energy during the wireless charging phase. While an increasing trend in the received energy with increasing number of users is observed when the network size is small, for larger network sizes, the total received energy is bounded by the physical constraints of the network such as maximum transmit power constraint and the number of antennas at the AP, and leads to diminishing received energy gain when the network size is further increased. 

\subsection{Effect of Transmission Power Allocation}
\setlength{\belowcaptionskip}{-20pt}
\begin{figure}[t]
\centering
\includegraphics[scale = 0.5]{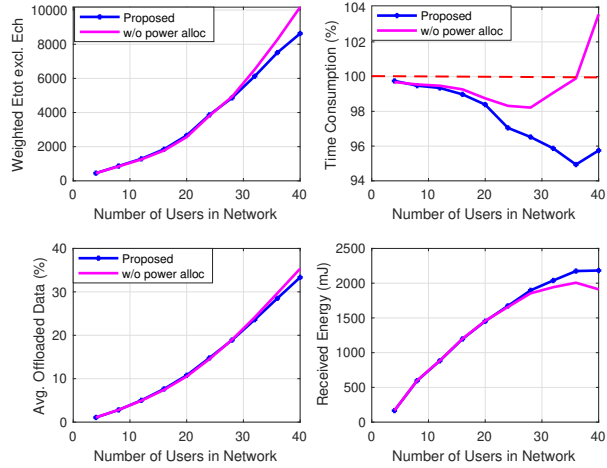}
\caption{Weighted total energy consumption excluding charging energy, percentage offloaded bits, total time consumption and total received energy with and without transmit power allocation}
\label{power_plot}
\end{figure}
Figure \ref{network_plot} also shows the effect of transmission power allocation in data transferring phases of computation offloading on the received energy as the network is increased, where each user requests $e_i = 500$mJ of energy and $u_i = 40$ kbits of data for computation. Using the proposed algorithm, we compare the energy consumption at the users and MEC-AP with and without transmission power control. In our proposed scheme, transmit power control is implemented indirectly through the time and data partitioning, given by (\ref{poweralloc}) via the optimized time allocation variables $\boldsymbol{t_u}$ and $\boldsymbol{t_d}$. In the scheme without power allocation, we fix the transmit power such that the MEC-AP allocates equal power for transmission beamforming to all users in downlink, and all users always use the maximum transmit power available. Note that this is the power allocation for data transmission in offloading and not to be mistaken with power allocation for the energy beamforming as discussed in Section VI. 

The top left plot shows the total weighted energy consumption excluding charging energy, the bottom left plot shows the percentage of offloaded bits, the top right plot shows the total time consumption and the bottom right plot shows the total received energy with and without transmit power allocation. We observe no difference in the performance for small network sizes, however, for large network, having no transmit power control leads to slightly more data offloaded to the MEC-AP, hence increasing the  energy consumption. This would also mean an increased wireless transmission time, and consequently we see lower received energy in this regime when transmission power allocation is not employed since wireless charging is only performed outside the wireless transmission time.

For the weighted total energy consumption, we exclude the energy consumption due to wireless charging to highlight the difference due to transmit power allocation in the data transmission phase only. An important finding for large network sizes is that without transmission power allocation, the network cannot cope with the data and energy requests within the latency constraint, and we see that the percentage time consumption exceeds 100\% for greater than 35 users in the network. Transmit power control is only consequential for larger network sizes and can be excluded from the optimization problem to reduce complexity in small networks. 

\section{Conclusion}
In this paper, we examined a massive MIMO enabled multi-access edge computing network providing computation offloading and on-request wireless charging to its connected users under a round trip latency constraint. We formulated a novel system-level problem to minimize the energy consumption for data offloading and maximize the received energy from wireless charging, and design efficient sequential algorithms to solve for data partitioning, time allocation and transmit energy beamforming matrices. Our algorithms demonstrated that data partitioning is a potent optimizing variable, as partial data offloading leads to significant reduction in system energy consumption, lower transmission times and consequentially higher amount of received charged energy at users. On the other hand, MEC-AP transmit power allocation for downlink data transmission has little effect on the system energy consumption for small network sizes.

Our algorithm also illustrated that even with significant amounts of data to be computed, the network can deliver decent amounts of charged energy to the users over an extended period of time, therefore validating a practical coexistence of computation offloading and wireless charging. A comparison with isotropic power transfer and equal power energy beamforming shows that optimal design of the energy beamforming directions and beam power allocation in wireless charging is crucial for energy efficiency, and is necessary for adopting on-request wireless charging as a billable service for future networks.


\section{Appendix}
\subsection{Appendix A - Proof for Lemma \ref{lemma1}}
Consider problem $(P_{\text{CO}})$ in (\ref{PseqCO}) at fixed values of $s_i$. The objective function is affine and convex. 
\begin{itemize}
\item Constraints (c), (e), (f), (h) for (P) in (\ref{PseqCO}) are linear. 
\item For constraints (a) and (b), the first terms are of the form $f(x) = x 2^{\frac{1}{x}}$ in $t_{u,i}$ and $t_{d,i}$ respectively, with $\nabla^2_x f(x) = \frac{2^{\frac{1}{x}}}{x^3}$ > 0 for $x > 0$, and hence $f(x)$ is convex in $x$. 
\item Relevant constraints are also linear and convex in $E_{u}$, $E_{m}$ and $T_j \ \forall j$. 
\end{itemize}
Based on the above, the objective is convex and all the constraints are convex in the remaining variables. Thus the problem is convex at given $s_i$.

\subsection{Appendix B - Proof for Theorem \ref{theorem5}}
We obtain matrix $\boldsymbol{C}$ from the Lagrangian for problem $(P_{\text{WC}})$ as follows
\setcounter{equation}{24}
\begin{align*}\label{Lsub2}
\mathcal{L}_{\text{WC}} &= \xi_i T_c \sum_{i=1}^{K} \xi_i \text{tr}(h_i^\ast \boldsymbol{W_q} h_i) T_c + \chi \left(\text{tr}(\boldsymbol{W_q}) - P \right) + \xi_i T_c \text{tr} \left(\left (\sum_{i=1}^{K} \rho_i h_i h_i^\ast  \right)\boldsymbol{W_q} \right) - \sum_{i=1}^{K} \rho_i e_i \\
&= \xi_i T_c \text{tr} \left( \left [ \chi \boldsymbol{I} +  \xi_i T_c \sum_{i=1}^{K} (1 + \rho_i) \boldsymbol{h_i h_i^\ast}\right ]\boldsymbol{W_q}\right) - \sum_{i=1}^{K} \rho_i e_i - \chi P \\
&= \xi_i T_c \text{tr} \left( \boldsymbol{CW_q}\right) - \sum_{i=1}^{K} \rho_i e_i - \chi P \tag{\theequation}
\end{align*}
The dual-function for the problem $(P_{\text{WC}})$ can be defined as
\begin{equation}\label{PWCdual}
g_{\text{WC}}(\boldsymbol{\rho}, \chi) = \max_{\boldsymbol{W_q}} \mathcal{L}_{\text{WC}} (\boldsymbol{W_q, \rho},\chi)
\end{equation}
Here we set $\boldsymbol{U_q = U_C}$ to maximize the Lagrangian $\mathcal{L}_{\text{WC}}$ such that, by applying the inequality relating trace of matrix product to the sum of eigenvalue products~\cite[Ch.~9, H.1.g.]{Olkin1979}, we have
\begin{equation}\label{trineq2}
\max_{\boldsymbol{W_q}} \ \text{tr}(\boldsymbol{CW_q}) =  \sum_{i = 1}^{N} \lambda_{C,i} \cdot \lambda_{q,i}
\end{equation}
where the eigenvalues of $\boldsymbol{C}$ and $\boldsymbol{W_q}$ are in the same descending order, $\lambda_{C,1} \geq, \ldots, \geq \lambda_{C,N}$, and $\lambda_{q,1} \geq, \ldots, \geq \lambda_{q,N}$, and therefore the sum of their eigenvalue products yields the maximum value for $\text{tr}(\boldsymbol{CW_q})$ in (\ref{trineq2}). The eigenvectors $\boldsymbol{U_C}$ are obtained based on the order of the corresponding eigenvalues in $\boldsymbol{\Lambda_C} = \text{diag}(\boldsymbol{\lambda_C})$. 

\subsection{Appendix C - Proof For Theorem \ref{theorem4}}
In the eigenvalue decomposition of $\boldsymbol{W_q^\star}$  as $\boldsymbol{W_q} = \boldsymbol{U_q \Lambda_q U_q^\ast}$, the diagonal matrix $\boldsymbol{\Lambda_q} \in \mathbb{R}^{N \times N}$ has power allocated across $K$ diagonal elements and the remaining eigenvalues for the $N-K$ beams is set to zero. Based on Theorem (\ref{theorem5}), equation (\ref{PseqWC}b) can be rewritten as
\begin{equation}
    \text{tr}(\boldsymbol{h}_i^\ast  \boldsymbol{U_q \Lambda_q U_q^\ast} \boldsymbol{h}_i) = \pi_i
\end{equation}
where $\pi_i = \frac{e_i}{\xi_i T_c} \ \forall i = 1...K$. We define the row vector, $\boldsymbol{r}_i^\ast = \boldsymbol{h}_i^\ast  \boldsymbol{U_q} = \boldsymbol{h}_i^\ast  \boldsymbol{U_B}$. Then 
\begin{equation}
    \text{tr}(\boldsymbol{r_i^\ast \Lambda_q r_i}) = \pi_i
\end{equation} 
Define row vector $\boldsymbol{d_i}^\ast = \textbf{diag}(\boldsymbol{r_i r_i^\ast}) \text{ for } i = 1...K$, matrix $\boldsymbol{D} \in \mathbb{R}^{K \times K} = [\boldsymbol{d_1^\ast}...\boldsymbol{d_K^\ast}]$, and vector $\boldsymbol{b} \in \mathbb{R}^{K \times 1} = [\pi_1 ... \pi_K]$. This results in constraint (\ref{P4}c) in ($P_{\text{BP}}$). The ordering of $\boldsymbol{\lambda_q}$ needs to be in the same order as $\boldsymbol{\lambda_C}$, that is, in descending order, so as to maximize (\ref{Lsub2}) as in (\ref{trineq2}) which gives us (\ref{P4}b) in ($P_{\text{BP}}$).

\bibliographystyle{./IEEEtran}
\bibliography{./wptbib}
\end{document}